\documentclass[a4paper, UKenglish, cleveref, autoref]{lipics-v2021}
\pdfoutput=1

\DeclareFontFamily{U}{shuffle}{}
\DeclareFontShape{U}{shuffle}{m}{n}{<-8>shuffle7 <8->shuffle10}{}

\usepackage[lipics, languages, complexity, semigroups]{Grossnathans_math}

\usepackage{hyperref}

\newcommand{\StVEsntl}{\mathbf{E}}
\newcommand{\V}{\FMVariety{V}}
\newcommand{\W}{\FMVariety{W}}
\newcommand{\FSVVjoinLI}{\V \FSVjoin \FSVLI}

\title{A Note on the Join of Varieties of Monoids with \texorpdfstring{$\FSVLI$}{LI}} %
\author{Nathan Grosshans}
       {Fachbereich Elektrotechnik/Informatik, University of Kassel, Germany \and \url{https://nathan.grosshans.me}}
       {nathan.grosshans@polytechnique.edu}
       {https://orcid.org/0000-0003-3400-1098}
       {}%
\authorrunning{N. Grosshans} %
\Copyright{Nathan Grosshans} %
\ccsdesc{Theory of computation~Formal languages and automata theory}
\ccsdesc{Theory of computation~Algebraic language theory}
\keywords{Varieties of monoids, join, LI} %

\acknowledgements
{
    I want to thank Thomas Place, who suggested the link between
    essentially-$\V$ stamps and $\V \FSVjoin \FSVLI$, but also Luc Segoufin who
    started the discussion with Thomas Place and encouraged me to write the
    present article.
    My thanks go as well to the anonymous referees for their helpful comments
    and suggestions.
    Finally, I want to mention that the introductions of Jean-Éric Pin's
    future book on algebraic automata theory and of Marc Zeitoun's works cited
    in the references have been helpful inspirations for my own introduction.
}

\nolinenumbers %

\hideLIPIcs  %

\EventEditors{Filippo Bonchi and Simon J. Puglisi}
\EventNoEds{2}
\EventLongTitle{46th International Symposium on Mathematical Foundations of Computer Science (MFCS 2021)}
\EventShortTitle{MFCS 2021}
\EventAcronym{MFCS}
\EventYear{2021}
\EventDate{August 23--27, 2021}
\EventLocation{Tallinn, Estonia}
\EventLogo{}
\SeriesVolume{202}
\ArticleNo{51}

\begin{document}

\maketitle

\begin{abstract}
In this note, we give a characterisation in terms of identities of the join of
$\V$ with the variety of finite locally trivial semigroups $\FSVLI$ for several
well-known varieties of finite monoids $\V$ by using classical
algebraic-automata-theoretic techniques. To achieve this, we use the new notion
of essentially-$\V$ stamps defined by Grosshans, McKenzie and Segoufin and show
that it actually coincides with the join of $\V$ and $\FSVLI$ precisely when
some natural condition on the variety of languages corresponding to $\V$ is
verified.

This work is a kind of rediscovery of the work of J.\ C.\ Costa around 20 years
ago from a rather different angle, since Costa's work relies on the use of
advanced developments in profinite topology, whereas what is presented here
essentially uses an algebraic, language-based approach.
 \end{abstract}

\section{Introduction}

One of the most fundamental problems in finite automata theory is the one of
\emph{characterisation}: given some subclass of the class of regular languages,
find out whether there is a way to characterise those languages using some class
of finite objects. This problem is often linked to and motivated by the problem
of \emph{decidability}: given some subclass of the class of regular languages,
find out whether there exists an algorithm testing the membership of any regular
language in that subclass.
The obvious approach to try to find a characterisation of a class of regular
languages would be to look for properties shared by all the minimal finite
automata of those languages. If we find such characterising properties, we can
then ask whether they can be checked by an algorithm to answer the problem of
decidability for this class of languages.
However, one of the most fruitful approaches of those two problems has been the
\emph{algebraic approach}, in which we basically replace automata with morphisms
into monoids: a language $L$ over an alphabet $\Sigma$ is then said to be
recognised by a morphism $\varphi$ into a monoid $M$ if and only if $L$ is the
inverse image by $\varphi$ of a subset of $M$. Under this notion of recognition,
each language has a minimal morphism recognising it, the \emph{syntatic
morphism} into the \emph{syntactic monoid} of that language, that are minimal
under some notion of division. The fundamental result on which this algebraic
approach relies is that a language is regular if and only if its syntactic
monoid is finite.
One can thus try to find a characterisation of some class of regular languages
by looking at the algebraic properties of the syntactic monoids of these
languages.

And many such characterisations that are decidable were indeed successfully
obtained since Schützenberger's seminal work in
1965~\cite{Schuetzenberger-1965}. His famous result, that really started the
field of \emph{algebraic automata theory}, states that the star-free regular
languages are exactly those whose syntactic monoids are finite and aperiodic.
Another important early result in that vein is the one of
Simon~\cite{Simon-1975} characterising the piecewise testable languages as
exactly those having a finite $\GreenJsymb$-trivial syntactic monoid.
Eilenberg~\cite{Books/Eilenberg-1976} was the first to prove that such algebraic
characterisations actually come as specific instances of a general bijective
correspondence between varieties of finite monoids and varieties of languages
-- classes of, respectively, finite monoids and regular languages closed under
natural operations. Thus, a class of regular languages can indeed be
characterised by the syntactic monoids of these languages, as soon as it
verifies some nice closure properties.
Eilenberg's result was later completed by Reiterman's
theorem~\cite{Reiterman-1982}, that uses a notion of identities defined using
profinite topology and states that a class of finite monoids is a variety of
finite monoids if and only if it is defined by a set of profinite identities.
Therefore, one can always characterise the variety of finite monoids associated
to a variety of languages by a set of profinite identities and, additionally,
this characterisation often leads to decidability, especially when this set is
finite.
A great deal of research works have been conducted to characterise varieties of
finite monoids or semigroups by profinite identities (see the book of
Almeida~\cite{Books/Almeida-1995} for an overview; see also the book chapter by
Pin~\cite{Pin-1997} for more emphasis on the ``language'' part).

A kind of varieties of finite monoids or semigroups that has attracted many
research efforts aiming for characterisations through identities are the
varieties defined as the join of two other varieties. Given two varieties of
finite monoids $\V$ and $\W$, the \emph{join of $\V$ and $\W$}, denoted by
$\V \FSVjoin \W$, is the least variety of finite monoids containing both $\V$
and $\W$. One of the main motivations to try to understand $\V \FSVjoin \W$ is
that the variety of languages corresponding to it by the Eilenberg
correspondence, $\DLang{\V \FSVjoin \W}$, is the one obtained by considering
direct products of automata recognising languages from both $\DLang{\V}$ and
$\DLang{\W}$, the varieties of languages corresponding to, respectively, $\V$
and $\W$. This is a fundamental operation on automata, and while it is
straightforward that $\DLang{\V \FSVjoin \W}$ is simply the least variety of
languages containing both $\DLang{\V}$ and $\DLang{\W}$, this does not at all
furnish a decidable characterisation of $\DLang{\V \FSVjoin \W}$, let alone a
set of identities defining $\V \FSVjoin \W$.
Generally speaking, the problem of finding a set of identities defining
$\V \FSVjoin \W$ is difficult (see~\cite{Books/Almeida-1995, Zeitoun-1996}): in
fact, there exist two varieties of finite semigroups that have a decidable
membership problem but whose join has an undecidable membership
problem~\cite{Albert-Baldinger-Rhodes-1992}.
However, sets of identities have been found for many specific joins: have a look
at~\cite{Almeida-1988, Almeida-Azevedo-1989, Azevedo-1990, Zeitoun-1995a,
Zeitoun-1995b, Azevedo-Zeitoun-1998, Costa-2001, Costa-2002} for some examples.

In this paper, we give a general method to find a set of identities defining the
join of an arbitrary variety of finite monoids $\V$ and the \emph{variety of
finite locally trivial semigroups $\FSVLI$}, as soon as one has a set of
identities defining $\V$ and $\V$ verifies some criterion. Joins of that sort
have been studied quite a lot in the literature we mentioned in the previous
paragraph (e.g.\ in~\cite{Azevedo-1990, Zeitoun-1995b, Costa-2001, Costa-2002}),
but while these works usually rely heavily on profinite topology with
some in-depth understanding of the structure of the elements of the so-called
free pro-$\V$ monoids and free pro-$\FSVLI$ semigroups, we present a method that
reduces the use of profinite topology to the minimum and that relies mainly on
algebraic and language-theoretic techniques.
The variety $\FSVLI$ is well-known to correspond to the class of languages for
which membership only depends on bounded-length prefixes and suffixes of words.
In~\cite{Grosshans-McKenzie-Segoufin-2021}, McKenzie, Segoufin and the author
introduced the notion of \emph{essentially-$\V$} stamps (surjective morphisms
$\varphi\colon \Sigma^* \to M$ for $\Sigma$ an alphabet and $M$ a finite monoid)
to characterise the built-in ability that programs over monoids in $\V$ have to
treat separately some constant-length beginning and ending of a word. Informally
said, a stamp is essentially-$\V$ when it behaves like a stamp into a monoid of
$\V$ as soon as a sufficiently long beginning and ending of the input word has
been fixed.
Our method builds on two results, that we prove in this article.
\begin{enumerate}
    \item
    \label{intro_result_1}
	The first result is a characterisation in terms of identities of the
	class $\StVEsntl\V$ of essentially-$\V$ stamps given a set of identities
	$E$ defining $\V$: a stamp is in $\StVEsntl\V$ if and only if it
	satisfies all identities
	$x^\omega y u z t^\omega = x^\omega y v z t^\omega$ for $u = v$ an
	identity in $E$ and where $x, y, z, t$ do appear neither in $u$ nor in
	$v$.
    \item
    \label{intro_result_2}
	The second result says that $\StVEsntl\V$ and $\FSVVjoinLI$ do coincide
	if and only if $\V$ verifies some criterion, that can be formulated in
	terms of quotient-expressibility in $\DLang{\V}$: any language
	$L \in \DLang{\V}$ must, for an arbitrary choice of $x, y$, be such that
	the quotient $u^{-1} L v^{-1}$ for $u$ and $v$ long enough can be
	expressed as the quotient $(x u)^{-1} K (v y)^{-1}$ for a
	$K \in \DLang{\V}$.
\end{enumerate}
Using these results, we can find a set of identities defining $\FSVVjoinLI$ as
soon as a set of identities defining $\V$ is known by proving that $\V$ verifies
the criterion in point~\ref{intro_result_2}. Note that actually, for technical
reasons, we work with the so-called \nem-variety of stamps corresponding to
$\FSVVjoinLI$ rather than directly with the variety of finite semigroups
$\FSVVjoinLI$, but this is not a problem since a variety of finite semigroups
can always be seen as an \nem-variety of stamps and vice versa.
We apply this method to reprove characterisations of the join of $\FSVLI$ with
each of the well-known varieties of finite monoids $\FMVR$, $\FMVL$, $\FMVJ$ and
any variety of finite groups.

The author noticed after proving those results that his work actually forms a
kind of rediscovery of the work of J.\ C.\ Costa in~\cite{Costa-2001}. He
defines an operator $U$ associating to each set of identities $E$ the exact same
new set $U(E)$ of identities as in point~\ref{intro_result_1}. Costa then
defines a property of cancellation for varieties of finite semigroups such that
for any $\V$ verifying it, $U(E)$ defines $\FSVVjoinLI$ for $E$ defining $\V$.
He finally uses this result to derive characterisations of $\FSVVjoinLI$ for all
the cases we are treating in our paper and many more.

What is, then, the contribution of our article? In a nutshell, it does mainly
use algebraic and language-theoretic techniques while Costa's work relies
heavily on profinite topology. In our setting, once the stage is set, all proofs
are quite straightforward without real difficulties and rely on classical
language-theoretic characterisations of the varieties under consideration. This
is to contrast with Costa's work, that for instance draws upon the difficult
analysis of the elements of free pro-$\FMVR$ monoids by Almeida and
Weil~\cite{Almeida-Weil-1997} to characterise $\FMVR \FSVjoin \FSVLI$.

\subparagraph{Organisation of the article.}
Section~\ref{sec:Preliminaries} is dedicated to the necessary preliminaries. In
Section~\ref{sec:Essentially-V}, we recall the definition of essentially-$\V$
stamps and prove the characterisation by identities of
point~\ref{intro_result_1} above.
Section~\ref{sec:Essentially-V_and_join_with_LI} is then dedicated to the
necessary and sufficient criterion for $\StVEsntl\V$ and $\FSVVjoinLI$ to
coincide presented in point~\ref{intro_result_2} and finally those results are
applied to specific cases in Section~\ref{sec:Applications}. We finish with a
short conclusion.

\section{Preliminaries}
\label{sec:Preliminaries}

We briefly introduce the mathematical material used in this paper.
For the basics and the classical results of automata theory, we refer the reader
to the two classical references of the domain by
Eilenberg~\cite{Books/Eilenberg-1974, Books/Eilenberg-1976} and
Pin~\cite{Books/Pin-1986}. For definitions and results specific to varieties of
stamps and associated profinite identities, see the articles by
Straubing~\cite{Straubing-2002} and by Pin and
Straubing~\cite{Pin-Straubing-2005}.
We also assume some basic knowledge of topology.

\subparagraph*{General notations.}
Let $i \in \N$ be a natural number. We shall denote by $[i]$ the set of all
natural numbers $n \in \N$ verifying $1 \leq n \leq i$.

\subparagraph*{Words and languages.}
Let $\Sigma$ be a finite alphabet. We denote by $\Sigma^*$ the set of all finite
words over $\Sigma$. We also denote by $\Sigma^+$ the set of all finite non
empty words over $\Sigma$, the empty word being denoted by $\emptyword$. Our
alphabets and words will always be finite, without further mention of this fact.
Given a word $w \in \Sigma^*$, we denote its length by $\length{w}$ and the set
of letters it contains by $\alphabet(w)$. Given $n \in \N$, we denote by
$\Sigma^{\geq n}$, $\Sigma^n$ and $\Sigma^{< n}$ the set of words over $\Sigma$
of length, respectively, at least $n$, exactly $n$ and less than $n$.

A \emph{language over $\Sigma$} is a subset of $\Sigma^*$. A language is
\emph{regular} if it is recognised by a deterministic finite automaton.
The \emph{quotient of a language $L$ over $\Sigma$ relative to the words $u$
and $v$ over $\Sigma$} is the language, denoted by $u^{-1}Lv^{-1}$, of the words
$w$ such that $uwv \in L$.

\subparagraph*{Monoids, semigroups and varieties.}
A \emph{semigroup} is a non-empty set equipped with an associative law that we
will write multiplicatively. A \emph{monoid} is a semigroup with an identity. An
example of a semigroup is $\Sigma^+$, the free semigroup over $\Sigma$.
Similarly $\Sigma^*$ is the free monoid over $\Sigma$.
A \emph{morphism $\varphi$ from a semigroup $S$ to a semigroup $T$} is a
function from $S$ to $T$ such that $\varphi(x y) = \varphi(x) \varphi(y)$ for
all $x, y \in S$. A morphism of monoids additionally requires that the identity
is preserved.
A semigroup $T$ is a \emph{subsemigroup} of a semigroup $S$ if $T$ is a subset
of $S$ and is equipped with the restricted law of $S$. Additionally the notion
of submonoids requires the presence of the identity.
A semigroup $T$ \emph{divides} a semigroup $S$ if $T$ is the image by a
semigroup morphism of a subsemigroup of $S$. Division of monoids is defined in
the same way.
The \emph{Cartesian (or direct) product} of two semigroups is simply the
semigroup given by the Cartesian product of the two underlying sets equipped
with the Cartesian product of their laws.
An element $s$ of a semigroup is \emph{idempotent} if $ss=s$.

A \emph{variety of finite monoids} is a non-empty class of finite monoids closed
under Cartesian product and monoid division. A \emph{variety of finite
semigroups} is defined similarly. When dealing with varieties, we consider only
finite monoids and semigroups, so we will drop the adjective finite when talking
about varieties in the rest of this article.

\subparagraph*{Varieties of stamps.}
Let $f\colon \Sigma^* \to \Gamma^*$ be a morphism from the free monoid over an
alphabet $\Sigma$ to the free monoid over an alphabet $\Gamma$, that we might
call an \allm-morphism. We say that $f$ is an \emph{\nem-morphism} (non-erasing
morphism) whenever $f(\Sigma) \subseteq \Gamma^+$.

We call \emph{stamp} a surjective morphism $\varphi\colon \Sigma^* \to M$ for
$\Sigma$ an alphabet and $M$ a finite monoid.
We say that a stamp $\varphi\colon \Sigma^* \to M$ \emph{\allm-divides}
(respectively \emph{\nem-divides}) a stamp $\psi\colon \Gamma^* \to N$ whenever
there exists an \allm-morphism (respectively \nem-morphism)
$f\colon \Sigma^* \to \Gamma^*$ and a surjective morphism
$\alpha\colon \Img(\psi \circ f) \to M$ such that
$\varphi = \alpha \circ \psi \circ f$.
The \emph{direct product} of two stamps $\varphi\colon \Sigma^* \to M$ and
$\psi\colon \Sigma^* \to N$ is the stamp
$\varphi\times\psi\colon \Sigma^* \to K$ such that $K$ is the submonoid of
$M \times N$ generated by $\set{(\varphi(a), \psi(a)) \mid a \in \Sigma}$ and
$\varphi\times\psi(a) = (\varphi(a), \psi(a))$ for all $a \in \Sigma$.

An \emph{\allm-variety of stamps} (respectively \emph{\nem-variety of stamps})
is a non-empty class of stamps closed under direct product and \allm-division
(respectively \nem-division).

We will often use the following characteristic index of stamps, defined
in~\cite{Chaubard-Pin-Straubing-2006b}. Consider a stamp
$\varphi\colon \Sigma^* \to M$. As $M$ is finite there is a $k \in \N_{>0}$ such
that $\varphi(\Sigma^{2k}) = \varphi(\Sigma^k)$: this implies that
$\varphi(\Sigma^k)$ is a semigroup. The least such $k$ is called the
\emph{stability index} of $\varphi$.

\subparagraph*{Varieties of languages.}
A language $L$ over an alphabet $\Sigma$ is \emph{recognised by a monoid $M$} if
there is a morphism $\varphi\colon \Sigma^* \to M$ and $F \subseteq M$ such that
$L = \varphi^{-1}(F)$. We also say that \emph{$\varphi$ recognises~$L$}. It is
well known that a language is regular if and only if it is recognised by a
finite monoid.
The \emph{syntactic congruence} of $L$, denoted by $\sim_L$, is the equivalence
relation on $\Sigma^*$ defined by $u \sim_L v$ for $u, v \in \Sigma^*$ whenever
for all $x, y \in \Sigma^*$, $xuy \in L$ if and only if $xvy \in L$. The
quotient $\Sigma^* \quotient \sim_L$ is a monoid, called \emph{the syntactic
monoid of $L$}, that recognises $L$ via the \emph{syntactic morphism $\eta_L$ of
$L$} sending any word $u$ to its equivalence class $[u]_{\sim_L}$ for $\sim_L$.
A stamp $\varphi\colon \Sigma^* \to M$ recognises $L$ if and only if there
exists a surjective morphism $\varphi\colon M \to \Sigma^* \quotient \sim_L$
verifying $\eta_L = \alpha \circ \varphi$.

A \emph{class of languages} $\mathcal{C}$ is a correspondence that associates a
set $\mathcal{C}(\Sigma)$ to each alphabet $\Sigma$.
A \emph{(\allm-)variety of languages} (respectively an \emph{\nem-variety of
languages}) $\LVariety{V}$ is a non-empty class of regular languages closed
under Boolean operations, quotients and inverses of \allm-morphisms
(respectively \nem-morphisms).
A classical result of
Eilenberg~\cite[Chapter VII, Section 3]{Books/Eilenberg-1976} says that there is
a bijective correspondence between varieties of monoids and varieties of
languages: to each variety of monoids $\V$ we can bijectively associate
$\DLang{\V}$ the variety of languages whose syntactic monoids belong to $\V$.
This was generalised by Straubing~\cite{Straubing-2002} to varieties of stamps:
to each \allm-variety (respectively \nem-variety) of stamps $\V$ we can
bijectively associate $\DLang{\V}$ the \allm-variety (respectively \nem-variety)
of languages whose syntactic morphisms belong to $\V$.
Given two \allm-varieties (respectively \nem-varieties) of stamps
$\StVariety{V_1}$ and $\StVariety{V_2}$, we have
$\StVariety{V_1} \subseteq \StVariety{V_2} \Leftrightarrow
 \DLang{\StVariety{V_1}} \subseteq \DLang{\StVariety{V_2}}$.

For $\V$ a variety of monoids, we define $\StVGen{\V}{\allm}$ the \allm-variety
of all stamps $\varphi\colon \Sigma^* \to M$ such that $M \in \V$. Of course, in
that case $\DLang{\V} = \DLang{\StVGen{\V}{\allm}}$.
Similarly, for $\V$ a variety of semigroups, we define $\StVGen{\V}{\nem}$ the
\nem-variety of all stamps $\varphi\colon \Sigma^* \to M$ such that
$\varphi(\Sigma^+) \in \V$. In that case, we consider $\DLang{\V}$ to be the
\nem-variety of languages corresponding to $\StVGen{\V}{\nem}$.
The operations $\StVGen{\cdot}{\allm}$ and $\StVGen{\cdot}{\nem}$ form bijective
correspondences between varieties of monoids and \allm-varieties of stamps and
between varieties of semigroups and \nem-varieties of stamps, respectively
(see~\cite{Straubing-2002}).

\subparagraph{Identities.}
Let $\Sigma$ be an alphabet. Given $u, v \in \Sigma^*$, we set
\[
    r(u, v) = \min\set{\card{M} \mid
		       \exists \varphi\colon \Sigma^* \to M \text{ stamp s.t. }
		       \varphi(u) \neq \varphi(v)}
\]
and $d(u, v) = 2^{-r(u, v)}$, using the conventions that
$\min\emptyset = +\infty$ and $2^{-\infty} = 0$. Then $d$ is a metric on
$\Sigma^*$.
The completion of the metric space $(\Sigma^*, d)$, denoted by
$(\widehat{\Sigma^*}, \widehat{d})$, is a metric monoid called the \emph{free
profinite monoid on $\Sigma^*$}. Its elements are all the formal limits
$\lim_{n \to \infty} x_n$ of Cauchy sequences $(x_n)_{n \geq 0}$ in
$(\Sigma^*, d)$ and the metric $d$ on $\Sigma^*$ extends to a metric
$\widehat{d}$ on $\widehat{\Sigma^*}$ defined by
$\widehat{d}(\lim_{n \to \infty} x_n, \lim_{n \to \infty} y_n) =
 \lim_{n \to \infty} d(x_n, y_n)$
for Cauchy sequences $(x_n)_{n \geq 0}$ and $(y_n)_{n \geq 0}$ in
$(\Sigma^*, d)$.
Note that, when it is clear from the context, we usually do not make the metric
explicit when talking about a metric space.
One important example of elements of $\widehat{\Sigma^*}$ is given by the
elements $x^\omega = \lim_{n \to \infty} x^{n!}$ for all $x \in \Sigma^*$.

Every finite monoid $M$ is considered to be a complete metric space equipped
with the discrete metric $d$ defined by
$d(m, n) =
 \begin{cases}
    0 & \text{if $m = n$}\\
    1 & \text{otherwise}
 \end{cases}$
for all $m, n \in M$.
Every stamp $\varphi\colon \Sigma^* \to M$ extends uniquely to a uniformly
continuous morphism $\widehat{\varphi}\colon \widehat{\Sigma^*} \to M$ with
$\widehat{\varphi}(\lim_{n \to \infty} x_n) = \lim_{n \to \infty} \varphi(x_n)$
for every Cauchy sequence $(x_n)_{n \geq 0}$ in $\Sigma^*$.
Similarly, every \allm-morphism $f\colon \Sigma^* \to \Gamma^*$ extends uniquely
to a uniformly continuous morphism
$\widehat{f}\colon \widehat{\Sigma^*} \to \widehat{\Gamma^*}$ with
$\widehat{f}(\lim_{n \to \infty} x_n) = \lim_{n \to \infty} f(x_n)$ for every
Cauchy sequence $(x_n)_{n \geq 0}$ in $\Sigma^*$.

For $u, v \in \widehat{A^*}$ with $A$ an alphabet, we say that a stamp
$\varphi\colon \Sigma^* \to M$ \emph{\allm-satisfies} (respectively
\emph{\nem-satisfies}) the identity $u = v$ if for every \allm-morphism
(respectively \nem-morphism) $f\colon A^* \to \Sigma^*$, it holds that
$\widehat{\varphi} \circ \widehat{f}(u) =
 \widehat{\varphi} \circ \widehat{f}(v)$.
Given a set of identities $E$, we denote by $\StVidentities{E}_\allm$
(respectively $\StVidentities{E}_\nem$) the class of stamps \allm-satisfying
(respectively \nem-satisfying) all the identities of $E$. When
$\StVidentities{E}_\allm$ (respectively $\StVidentities{E}_\nem$) is equal to an
\allm-variety (respectively \nem-variety) of stamps $\V$, we say that $E$
\emph{\allm-defines} (respectively \emph{\nem-defines}) $\V$.

\begin{theorem}[{\cite[Theorem 2.1]{Pin-Straubing-2005}}]
    A class of stamps is an \allm-variety (respectively \nem-variety) of stamps
    if and only if it can be \allm-defined (respectively \nem-defined) by a set
    of identities.
\end{theorem}

To give some examples, the classical varieties of monoids $\FMVJ$, $\FMVR$ and
$\FMVL$ can be characterised by identities in the following way:
\begin{align*}
    \StVGen{\FMVR}{\allm} &=
    \StVidentities{(ab)^\omega a = (ab)^\omega}_\allm =
    \StVidentities{(ab)^\omega a = (ab)^\omega}_\nem\\
    \StVGen{\FMVL}{\allm} &=
    \StVidentities{b (ab)^\omega = (ab)^\omega}_\allm =
    \StVidentities{b (ab)^\omega = (ab)^\omega}_\nem\\
    \StVGen{\FMVJ}{\allm} &=
    \StVidentities{(ab)^\omega a = (ab)^\omega,
		   b (ab)^\omega = (ab)^\omega}_\allm =
    \StVidentities{(ab)^\omega a = (ab)^\omega,
		   b (ab)^\omega = (ab)^\omega}_\nem
    \displaypunct{.}
\end{align*}

\subparagraph{Finite locally trivial semigroups and the join operation.}
The variety $\FSVLI$ of finite locally trivial semigroups is well-known to
verify
$\StVGen{\FSVLI}{\nem} = \StVidentities{x^\omega y x^\omega = x^\omega}_\nem$
and to be such that for any alphabet $\Sigma$, the set $\DLang{\FSVLI}(\Sigma)$
consists of all Boolean combinations of languages of the form $u \Sigma^*$ or
$\Sigma^* u$ for $u \in \Sigma^*$, or equivalently of all languages of the form
$U \Sigma^* V \cup W$ with $U, V, W \subseteq \Sigma^*$ finite
(see~\cite[p.~38]{Books/Pin-1986}).

Given a variety of monoids $\V$, the join of $\V$ and $\FSVLI$, denoted by
$\V \FSVjoin \FSVLI$, is the inclusion-wise least variety of semigroups
containing both $\V$ and $\FSVLI$. In fact, a finite semigroup $S$ belongs to
$\V \FSVjoin \FSVLI$ if and only if there exist $M \in \V$ and $T \in \FSVLI$
such that $S$ divides the semigroup $M \times T$.
(See~\cite[Chapter V, Exercise 1.1]{Books/Eilenberg-1976}.)
We can prove the following adaptation to \nem-varieties of the classical results
about joins (see the appendix for the proof).

\begin{proposition}
\label{ptn:V_join_LI_ne-variety}
    Let $\V$ be a variety of monoids.
    Then $\StVGen{\FSVVjoinLI}{\nem}$ is the inclusion-wise least \nem-variety
    of stamps containing both $\StVGen{\V}{\allm}$ and $\StVGen{\FSVLI}{\nem}$.
    Moreover, $\DLang{\FSVVjoinLI}$ is the inclusion-wise least \nem-variety of
    languages containing both $\DLang{\V}$ and $\DLang{\FSVLI}$ and verifies
    that $\DLang{\FSVVjoinLI}(\Sigma)$ is the Boolean closure of
    $\DLang{\V}(\Sigma) \cup \DLang{\FSVLI}(\Sigma)$ for each alphabet $\Sigma$.
\end{proposition}

\section{Essentially-\texorpdfstring{$\V$}{V} stamps}
\label{sec:Essentially-V}

In this section, we give a characterisation of essentially-$\V$ stamps (first
defined in~\cite{Grosshans-McKenzie-Segoufin-2021}), for $\V$ a variety of
monoids, in terms of identities. We first recall the definition.

\begin{definition}
    Let $\V$ be a variety of monoids.
    Let $\varphi\colon \Sigma^* \to M$ be a stamp and let $s$ be its stability
    index.

    We say that $\varphi$ is \emph{essentially-$\V$} whenever there exists a
    stamp $\mu\colon \Sigma^* \to N$ with $N \in \V$ such that for all
    $u, v \in \Sigma^*$, we have
    \[
        \mu(u) = \mu(v) \Rightarrow
        \bigl(\varphi(x u y) = \varphi(x v y) \quad
              \forall x, y \in \Sigma^s\bigr)
        \displaypunct{.}
    \]
    We will denote by $\StVEsntl\V$ the class of all essentially-$\V$ stamps.%
    \footnote{Essentially-$\V$ stamps are called that way by analogy with
    quasi-$\V$ stamps and the class of essentially-$\V$ stamps is denoted by
    $\StVEsntl\V$ by analogy with $\StVQuasi\V$, the notation for the class of
    quasi-$\V$ stamps. This makes sense since the initial motivation for the
    definition of essentially-$\V$ stamps was to capture the class of stamps
    into monoids of $\V$ that have the additional ability to treat separately
    some constant-length beginning and ending of a word. This ability can indeed
    be seen as orthogonal to the additional ability of stamps into monoids in
    $\V$ to perform modular counting on the positions of letters in a word,
    which is often handled by considering quasi-$\V$ stamps.
    (See~\cite{Grosshans-McKenzie-Segoufin-2021} for more.)
    Our definition of $\StVEsntl\V$ does unfortunately not coincide with the
    usual definition of $\StVEsntl\V$, that classically denotes the variety of
    monoids $M$ such that the submonoid generated by the idempotents of $M$ is
    in $\V$. (This comes, among others, from the fact that the obtained variety
    of monoids does always contain at least all finite groups.)}
\end{definition}

Now, we give a characterisation for a stamp to be essentially-$\V$, based on a
specific congruence depending on that stamp.

\begin{definition}
    Let $\varphi\colon \Sigma^* \to M$ be a stamp and let $s$ be its stability
    index.
    We define the equivalence relation $\equiv_\varphi$ on $\Sigma^*$ by
    $u \equiv_\varphi v$ for $u, v \in \Sigma^*$ whenever
    $\varphi(x u y) = \varphi(x v y)$ for all $x, y \in \Sigma^{\geq s}$.
\end{definition}

\begin{proposition}
\label{ptn:Essential_congruence}
    Let $\varphi\colon \Sigma^* \to M$ be a stamp.
    Then $\equiv_\varphi$ is a congruence of finite index and for any variety of
    monoids $\V$, we have $\varphi \in \StVEsntl\V$ if and only if
    $\Sigma^* \quotient \equiv_\varphi \in \V$.
\end{proposition}

\begin{proof}
    Let us denote by $s$ the stability index of $\varphi$.

    The equivalence relation $\equiv_\varphi$ is a congruence because given
    $u, v \in \Sigma^*$ verifying $u \equiv_\varphi v$, for all
    $\alpha, \beta \in \Sigma^*$, we have
    $\alpha u \beta \equiv_\varphi \alpha v \beta$ since for any
    $x, y \in \Sigma^{\geq s}$, it holds that
    $\varphi(x \alpha u \beta y) = \varphi(x \alpha v \beta y)$ because
    $x \alpha, \beta y \in \Sigma^{\geq s}$.
    Furthermore, this congruence is of finite index because for all
    $u, v \in \Sigma^*$, we have that $\varphi(u) = \varphi(v)$ implies
    $u \equiv_\varphi v$.

    Let now $\V$ be a variety of monoids.
    Assume first that $\Sigma^* \quotient \equiv_\varphi \in \V$. It is quite
    direct to see that $\varphi \in \StVEsntl\V$, as the stamp
    $\mu\colon \Sigma^* \to \Sigma^* \quotient \equiv_\varphi$ defined by
    $\mu(w) = [w]_{\equiv_\varphi}$ for all $w \in \Sigma^*$ witnesses this
    fact.
    Assume then that $\varphi \in \StVEsntl\V$. This means that there exists a
    stamp $\mu\colon \Sigma^* \to N$ with $N \in \V$ such that for all
    $u, v \in \Sigma^*$, we have
    \[
        \mu(u) = \mu(v) \Rightarrow
        \bigl(\varphi(x u y) = \varphi(x v y) \quad
              \forall x, y \in \Sigma^s\bigr)
        \displaypunct{.}
    \]
    Now consider $u, v \in \Sigma^*$ such that $\mu(u) = \mu(v)$. For any
    $x, y \in \Sigma^{\geq s}$, we have that $x = x_1 x_2$ with
    $x_1 \in \Sigma^*$ and $x_2 \in \Sigma^s$ as well as $y = y_1 y_2$ with
    $y_1 \in \Sigma^s$ and $y_2 \in \Sigma^*$, so that
    $\varphi(x u y) = \varphi(x_1) \varphi(x_2 u y_1) \varphi(y_2) =
     \varphi(x_1) \varphi(x_2 v y_1) \varphi(y_2) = \varphi(x v y)$.
    Hence, $u \equiv_\varphi v$.
    Therefore, for all $u, v \in \Sigma^*$, we have that $\mu(u) = \mu(v)$
    implies $u \equiv_\varphi v$, so we can define the mapping
    $\alpha\colon N \to \Sigma^* \quotient \equiv_\varphi$ such that
    $\alpha(\mu(w)) = [w]_{\equiv_\varphi}$ for all $w \in \Sigma^*$. It is easy
    to check that $\alpha$ is actually a surjective morphism. Thus, we can
    conclude that $\Sigma^* \quotient \equiv_\varphi$, which divides $N$,
    belongs to $\V$.
\end{proof}

Using this characterisation, we prove that given a set of identities
\nem-defining $\StVGen{\V}{\allm}$ for a variety of monoids $\V$, we get a set
of identities \nem-defining $\StVEsntl\V$.

\begin{proposition}
\label{ptn:Essentially-V_identities}
    Let $\V$ be a variety of monoids and let $E$ be a set of identities such
    that $\StVGen{\V}{\allm} = \StVidentities{E}_\nem$.
    Then $\StVEsntl\V$ is an \nem-variety of stamps and
    \[
	\StVEsntl\V =
	\StVidentities{x^\omega y u z t^\omega = x^\omega y v z t^\omega \mid
		       u = v \in E,
		       x, y, z, t \notin \alphabet(u) \cup \alphabet(v)}_\nem
	\displaypunct{.}
    \]
\end{proposition}

\begin{proof}
    Let
    \[
	F = \set{x^\omega y u z t^\omega = x^\omega y v z t^\omega \mid
		 u = v \in E, x, y, z, t \notin \alphabet(u) \cup \alphabet(v)}
	\displaypunct{.}
    \]

    Central to the proof is the following claim.
    \begin{claim}
    \label{clm:Essential_congruence}
	Let $\varphi\colon \Sigma^* \to M$ be a stamp.
	Consider the stamp
	$\mu\colon \Sigma^* \to \Sigma^* \quotient \equiv_\varphi$ defined by
	$\mu(w) = [w]_{\equiv_\varphi}$ for all $w \in \Sigma^*$.
	It holds that for all $u, v \in \widehat{\Sigma^*}$,
	\[
	    \widehat{\mu}(u) = \widehat{\mu}(v) \Leftrightarrow
	    \bigl(\widehat{\varphi}(\alpha^\omega \beta u \gamma \delta^\omega)
		  =
		  \widehat{\varphi}(\alpha^\omega \beta v \gamma \delta^\omega)
		  \quad
		  \forall \alpha, \beta, \gamma, \delta \in \Sigma^+
	    \bigr)
	    \displaypunct{.}
	\]
    \end{claim}

    Before we prove Claim~\ref{clm:Essential_congruence}, we use it to prove
    that $\StVEsntl\V = \StVidentities{F}_\nem$.

    \proofsubparagraph{Inclusion from left to right.}
    Let $\varphi\colon \Sigma^* \to M$ be a stamp in $\StVEsntl\V$.
    Consider the stamp
    $\mu\colon \Sigma^* \to \Sigma^* \quotient \equiv_\varphi$ defined by
    $\mu(w) = [w]_{\equiv_\varphi}$ for all $w \in \Sigma^*$. Since
    $\varphi \in \StVEsntl\V$, Proposition~\ref{ptn:Essential_congruence} tells
    us that $\Sigma^* \quotient \equiv_\varphi \in \V$, hence
    $\mu \in \StVGen{\V}{\allm}$.

    Let us consider any identity
    $x^\omega y u z t^\omega = x^\omega y v z t^\omega \in F$. It is written on
    an alphabet $B$ that is the union of the alphabet $A$ on which $u = v \in E$
    is written and of $x, y, z, t \in B \setminus A$.
    Let $f\colon B^* \to \Sigma^*$ be an \nem-morphism.
    Since $\mu \in \StVGen{\V}{\allm}$, we have that $\mu$ \nem-satisfies the
    identity $u = v$, so that
    $\widehat{\mu}(\widehat{f}(u)) = \widehat{\mu}(\widehat{f}(v))$. Notice that
    we have that $\widehat{f}(x^\omega) = f(x)^\omega$ as well as
    $\widehat{f}(t^\omega) = f(t)^\omega$ and that
    $f(x), f(y), f(z), f(t) \in \Sigma^+$ because $f$ is non-erasing.
    Therefore, we have
    \begin{align*}
	\widehat{\varphi}\bigl(\widehat{f}(x^\omega y u z t^\omega)\bigr)
	& = \widehat{\varphi}\bigl(f(x)^\omega f(y) \widehat{f}(u)
				   f(z) f(t)^\omega\bigr)\\
	& = \widehat{\varphi}\bigl(f(x)^\omega f(y) \widehat{f}(v)
				   f(z) f(t)^\omega\bigr)\\
	& = \widehat{\varphi}\bigl(\widehat{f}(x^\omega y v z t^\omega)\bigr)
    \end{align*}
    by Claim~\ref{clm:Essential_congruence}.
    As this holds for any \nem-morphism $f\colon B^* \to \Sigma^*$, we can
    conclude that $\varphi$ \nem-satisfies the identity
    $x^\omega y u z t^\omega = x^\omega y v z t^\omega$.

    This is true for any identity in $F$, so
    $\varphi \in \StVidentities{F}_\nem$.
    In conclusion, $\StVEsntl\V \subseteq \StVidentities{F}_\nem$.

    \proofsubparagraph{Inclusion from right to left.}
    Let $\varphi\colon \Sigma^* \to M$ be a stamp in $\StVidentities{F}_\nem$.
    Consider the stamp
    $\mu\colon \Sigma^* \to \Sigma^* \quotient \equiv_\varphi$ defined by
    $\mu(w) = [w]_{\equiv_\varphi}$ for all $w \in \Sigma^*$.
    We are now going to show that $\mu \in \StVGen{\V}{\allm}$.

    Take any identity $u = v \in E$ written on an alphabet $A$. There exists an
    identity $x^\omega y u z t^\omega = x^\omega y v z t^\omega \in F$ written
    on an alphabet $B$ such that $A \subseteq B$ and
    $x, y, z, t \in B \setminus A$.
    Let $f\colon A^* \to \Sigma^*$ be an \nem-morphism.

    Take any $\alpha, \beta, \gamma, \delta \in \Sigma^+$. Let us define the
    \nem-morphism $g\colon B^* \to \Sigma^*$ as the unique one which extends $f$
    by letting $g(x) = \alpha$, $g(y) = \beta$, $g(z) = \gamma$ and
    $g(t) = \delta$.
    Observe in particular that $\widehat{g}(w) = \widehat{f}(w)$ for any
    $w \in \widehat{A^*}$ and that
    $\widehat{g}(x^\omega) = g(x)^\omega = \alpha^\omega$ as well as
    $\widehat{g}(t^\omega) = \delta^\omega$.
    Now, as $\varphi$ \nem-satisfies
    $x^\omega y u z t^\omega = x^\omega y v z t^\omega$, we have that
    \[
	\widehat{\varphi}\bigl(\alpha^\omega \beta \widehat{f}(u)
			       \gamma \delta^\omega\bigr)
	= \widehat{\varphi}\bigl(\widehat{g}(x^\omega y u z t^\omega)\bigr)
	= \widehat{\varphi}\bigl(\widehat{g}(x^\omega y v z t^\omega)\bigr)
	= \widehat{\varphi}\bigl(\alpha^\omega \beta \widehat{f}(v) 
				 \gamma \delta^\omega\bigr)
	\displaypunct{.}
    \]
    Since this holds for any $\alpha, \beta, \gamma, \delta \in \Sigma^+$, by
    Claim~\ref{clm:Essential_congruence}, we have that
    $\widehat{\mu}(\widehat{f}(u)) = \widehat{\mu}(\widehat{f}(v))$.

    Therefore, $\widehat{\mu}(\widehat{f}(u)) = \widehat{\mu}(\widehat{f}(v))$
    for any \nem-morphism $f\colon A^* \to \Sigma^*$, which means that $\mu$
    \nem-satisfies $u = v$.

    Since this holds for any $u = v \in E$, we have that
    $\mu \in \StVGen{\V}{\allm}$, which implies that
    $\Sigma^* \quotient \equiv_\varphi \in \V$ and thus
    $\varphi \in \StVEsntl\V$ by Proposition~\ref{ptn:Essential_congruence}. In
    conclusion, $\StVidentities{F}_\nem \subseteq \StVEsntl\V$.

    \medskip

    The claim still needs to be proved.

    \begin{claimproof}[Proof of Claim~\ref{clm:Essential_congruence}]
	Let $\varphi\colon \Sigma^* \to M$ be a stamp of stability index $s$.
	Consider the stamp
	$\mu\colon \Sigma^* \to \Sigma^* \quotient \equiv_\varphi$ defined by
	$\mu(w) = [w]_{\equiv_\varphi}$ for all $w \in \Sigma^*$.
	We now want to show that for all $u, v \in \widehat{\Sigma^*}$,
	\[
	    \widehat{\mu}(u) = \widehat{\mu}(v) \Leftrightarrow
	    \bigl(\widehat{\varphi}(\alpha^\omega \beta u \gamma \delta^\omega)
		  =
		  \widehat{\varphi}(\alpha^\omega \beta v \gamma \delta^\omega)
		  \quad
		  \forall \alpha, \beta, \gamma, \delta \in \Sigma^+
	    \bigr)
	    \displaypunct{.}
	\]

	Let $u, v \in \widehat{\Sigma^*}$. There exist two Cauchy sequences
	$(u_n)_{n \geq 0}$ and $(v_n)_{n \geq 0}$ in $\Sigma^*$ such that
	$u = \lim_{n \to \infty} u_n$ and $v = \lim_{n \to \infty} v_n$.
	As $\Sigma^* \quotient \equiv_\varphi$ and $M$ are discrete, we have
	that all four Cauchy sequences $\bigl(\mu(u_n)\bigr)_{n \geq 0}$,
	$\bigl(\varphi(u_n)\bigr)_{n \geq 0}$, $\bigl(\mu(v_n)\bigr)_{n \geq 0}$
	and $\bigl(\varphi(v_n)\bigr)_{n \geq 0}$ are ultimately constant. So
	there exists $k \in \N$ such that $\widehat{\mu}(u) = \mu(u_k)$,
	$\widehat{\varphi}(u) = \varphi(u_k)$, $\widehat{\mu}(v) = \mu(v_k)$ and
	$\widehat{\varphi}(v) = \varphi(v_k)$.

	Assume first that $\widehat{\mu}(u) = \widehat{\mu}(v)$.
	Take any $\alpha, \beta, \gamma, \delta \in \Sigma^+$.
	Since $M$ is discrete, both Cauchy sequences
	$\bigl(\varphi(\alpha^{n!})\bigl)_{n \geq 0}$ and
	$\bigl(\varphi(\delta^{n!})\bigl)_{n \geq 0}$ are ultimately constant.
	So there exists $l \in \N$ such that for all $m \in \N, m \geq l$, we
	have $\widehat{\varphi}(\alpha^\omega) = \varphi(\alpha^{m!})$ and
	$\widehat{\varphi}(\delta^\omega) = \varphi(\delta^{m!})$.
	Hence, taking $m \in \N, m \geq l$ such that
	$\length{\alpha^{m!} \beta} \geq s$ and
	$\length{\gamma \delta^{m!}} \geq s$, it follows that
	\[
	    \widehat{\varphi}(\alpha^\omega \beta u \gamma \delta^\omega) =
	    \varphi(\alpha^{m!} \beta u_k \gamma \delta^{m!}) =
	    \varphi(\alpha^{m!} \beta v_k \gamma \delta^{m!}) =
	    \widehat{\varphi}(\alpha^\omega \beta v \gamma \delta^\omega)
	\]
	because
	$[u_k]_{\equiv_\varphi} = \widehat{\mu}(u) =
	 \widehat{\mu}(v) = [v_k]_{\equiv_\varphi}$.
	Thus, we have that
	\[
	    \widehat{\varphi}(\alpha^\omega \beta u \gamma \delta^\omega) =
	    \widehat{\varphi}(\alpha^\omega \beta v \gamma \delta^\omega)
	\]
	for all $\alpha, \beta, \gamma, \delta \in \Sigma^+$.

	Assume then that
	$\widehat{\varphi}(\alpha^\omega \beta u \gamma \delta^\omega) =
	 \widehat{\varphi}(\alpha^\omega \beta v \gamma \delta^\omega)$
	for all $\alpha, \beta, \gamma, \delta \in \Sigma^+$.
	Take any $\alpha, \beta \in \Sigma^{\geq s}$. Since $\varphi(\Sigma^s)$
	is a finite semigroup and verifies that
	$\varphi(\Sigma^s) = \varphi(\Sigma^s)^2$, by a classical result in
	finite semigroup theory (see
	e.g.~\cite[Chapter~1, Proposition~1.12]{Books/Pin-1986}), we have that
	there exist $\alpha_1, e, f, \beta_2 \in \Sigma^s$ and
	$\alpha_2, \beta_1 \in \Sigma^{\geq s}$ such that
	$\varphi(\alpha_1 e \alpha_2) = \varphi(\alpha)$ and
	$\varphi(\beta_1 f \beta_2) = \varphi(\beta)$ with $\varphi(e)$ and
	$\varphi(f)$ idempotents.
	Now, since $\varphi(e)$ is idempotent, we have that
	\[
	    \widehat{\varphi}(e^\omega) =
	    \widehat{\varphi}(\lim_{n \to \infty} e^{n!}) =
	    \lim_{n \to \infty} \varphi(e^{n!}) =
	    \lim_{n \to \infty} \varphi(e)^{n!} =
	    \varphi(e)
	\]
	and similarly, $\widehat{\varphi}(f^\omega) = \varphi(f)$.
	So it follows that
	\begin{align*}
	    \varphi(\alpha u_k \beta)
	    & = \varphi(\alpha_1 e \alpha_2 u_k \beta_1 f \beta_2)\\
	    & = \widehat{\varphi}(\alpha_1 e^\omega \alpha_2 u
				  \beta_1 f^\omega \beta_2)\\
	    & = \widehat{\varphi}(\alpha_1 e^\omega \alpha_2 v
				  \beta_1 f^\omega \beta_2)\\
	    & = \varphi(\alpha_1 e \alpha_2 v_k \beta_1 f \beta_2)\\
	    & = \varphi(\alpha v_k \beta)
	    \displaypunct{.}
	\end{align*}
	As this is true for any $\alpha, \beta \in \Sigma^{\geq s}$, by
	definition it holds that $u_k \equiv_\varphi v_k$, hence
	$\widehat{\mu}(u) = \mu(u_k) = \mu(v_k) = \widehat{\mu}(v)$.
    \end{claimproof}

    This concludes the proof of the proposition.
\end{proof}

\section{Essentially-\texorpdfstring{$\V$}{V} stamps and the join of
	 \texorpdfstring{$\V$}{V} and \texorpdfstring{$\FSVLI$}{LI}}
\label{sec:Essentially-V_and_join_with_LI}

In this section, we establish the link between essentially-$\V$ stamps and
$\V \FSVjoin \FSVLI$ and give a criterion that characterises exactly when they
do correspond.

More precisely, consider the following criterion for a variety of monoids $\V$.
\theoremstyle{definition}
\newtheorem{criterion}{Criterion}
\renewcommand{\thecriterion}{(\Alph{criterion})}
\begin{criterion}
\label{ctn:variety_criterion}
    For any $L \in \DLang{\V}(\Sigma)$ with $\Sigma$ an alphabet, we have
    $x L y \in \DLang{\V \FSVjoin \FSVLI}(\Sigma)$ for all $x, y \in \Sigma^*$.
\end{criterion}
It is a kind of mild closure condition that appears to be a sufficient and
necessary condition for $\StVEsntl\V$ and $\V \FSVjoin \FSVLI$ to correspond.

\begin{proposition}
\label{ptn:Essentially-V-V_join_LI}
    Let $\V$ be a variety of monoids.
    Then $\StVGen{\V \FSVjoin \FSVLI}{\nem} \subseteq \StVEsntl\V$ and equality
    holds if and only if $\V$ verifies criterion~\ref{ctn:variety_criterion}.
\end{proposition}

Why this proposition is useful to give characterisations of $\FSVVjoinLI$ in
terms of identities will become clear in the next section. For now, we focus on
its proof, that entirely relies on the following characterisation of the
languages recognised by essentially-$\V$ stamps.

\begin{proposition}
\label{ptn:Essentially-V_languages}
    Let $\V$ be a variety of monoids.
    For any alphabet $\Sigma$, the set $\DLang{\StVEsntl\V}(\Sigma)$ consists of
    all Boolean combinations of languages of the form $x L y$ for
    $L \in \DLang{\V}(\Sigma)$ and $x, y \in \Sigma^*$.
\end{proposition}

\begin{proof}
    Let $\LVariety{C}$ be the class of languages such that for any alphabet
    $\Sigma$, the set $\LVariety{C}(\Sigma)$ consists of all Boolean
    combinations of languages of the form $x L y$ for $L \in \DLang{\V}(\Sigma)$
    and $x, y \in \Sigma^*$.

    Let $\Sigma$ be an alphabet. We need to show that
    $\DLang{\StVEsntl\V}(\Sigma) = \LVariety{C}(\Sigma)$.

    \proofsubparagraph{Inclusion from right to left.}
    Let $L \in \DLang{\V}(\Sigma)$ and $x, y \in \Sigma^*$.
    Let $\mu\colon \Sigma^* \to N$ be the syntactic morphism of $L$: this
    implies that $N \in \V$ and that there exists $F \subseteq N$ such that
    $L = \mu^{-1}(F)$.
    Let also $\varphi\colon \Sigma^* \to M$ be the syntactic morphism of the
    language
    $x L y =
     x \Sigma^* y \cap \Sigma^{\length{x}} \mu^{-1}(F) \Sigma^{\length{y}}$
    and let $s$ be its stability index.
    We then consider $u, v \in \Sigma^*$ such that $\mu(u) = \mu(v)$. Take any
    $x', y' \in \Sigma^*$ such that $\length{x'} \geq \length{x}$ and
    $\length{y'} \geq \length{y}$.
    We clearly have that $x' u y' \in x \Sigma^* y$ if and only if
    $x' v y' \in x \Sigma^* y$. Moreover, $x' = x'_1 x'_2$ for
    $x'_1 \in \Sigma^{\length{x}}$ and $x'_2 \in \Sigma^*$ and $y' = y'_1 y'_2$
    for $y'_1 \in \Sigma^*$ and $y'_2 \in \Sigma^{\length{y}}$, so that
    \begin{align*}
	x' u y' \in \Sigma^{\length{x}} \mu^{-1}(F) \Sigma^{\length{y}}
	& \Leftrightarrow \mu(x'_2 u y'_1) \in F\\
	& \Leftrightarrow \mu(x'_2 v y'_1) \in F\\
	& \Leftrightarrow x' v y' \in
	  \Sigma^{\length{x}} \mu^{-1}(F) \Sigma^{\length{y}}
	\displaypunct{.}
    \end{align*}
    Hence, $x' u y' \in x L y$ if and only if $x' v y' \in x L y$ for all
    $x', y' \in \Sigma^*$ such that $\length{x'} \geq \length{x}$ and
    $\length{y'} \geq \length{y}$, so that, by definition of the stability index
    $s$ of $\varphi$ and as $\varphi$ is the syntactic morphism of $x L y$, we
    have $\varphi(x' u y') = \varphi(x' v y')$ for all $x', y' \in \Sigma^s$.
    Thus, it follows that $\varphi \in \StVEsntl\V$.

    This implies that $x L y \in \DLang{\StVEsntl\V}(\Sigma)$. Therefore, since
    this is true for any $L \in \DLang{\V}(\Sigma)$ and $x, y \in \Sigma^*$ and
    since $\DLang{\StVEsntl\V}(\Sigma)$ is closed under Boolean operations, we
    can conclude that
    $\LVariety{C}(\Sigma) \subseteq \DLang{\StVEsntl\V}(\Sigma)$.

    \proofsubparagraph{Inclusion from left to right.}
    Let $L \in \DLang{\StVEsntl\V}(\Sigma)$ and let
    $\varphi\colon \Sigma^* \to M$ be its syntactic morphism: it is an
    essentially-$\V$ stamp. Given $s$ its stability index, this means there
    exists a stamp $\mu\colon \Sigma^* \to N$ with $N \in \V$ such that for all
    $u, v \in \Sigma^*$, we have
    \[
	\mu(u) = \mu(v) \Rightarrow
	\bigl(\varphi(x u y) = \varphi(x v y) \quad
	      \forall x, y \in \Sigma^s\bigr)
	\displaypunct{.}
    \]

    For each $m \in N$ and $x, y \in \Sigma^s$ consider the language
    $x \mu^{-1}(m) y$. For any two words $w, w' \in x \mu^{-1}(m) y$, we have
    $w = x u y$ and $w' = x v y$ with $\mu(u) = \mu(v) = m$, so that
    $\varphi(w) = \varphi(w')$. By definition of the syntactic morphism, this
    means that for all $m \in N$ and $x, y \in \Sigma^s$, either
    $x \mu^{-1}(m) y \subseteq L$ or $x \mu^{-1}(m) y \cap L = \emptyset$.
    Therefore, there exists a set
    $E \subseteq N \times \Sigma^s \times \Sigma^s$ such that
    $L \cap \Sigma^{\geq 2 s} = \bigcup_{(m, x, y) \in E} x \mu^{-1}(m) y$,
    hence
    \[
	L = \bigcup_{(m, x, y) \in E} x \mu^{-1}(m) y \cup F
    \]
    for a certain $F \subseteq \Sigma^{< 2 s}$.

    Take $w \in F$. We have that
    $\set{w} =
     w \Sigma^* \cap \bigcap_{a \in \Sigma} (\Sigma^* \setminus w a \Sigma^*)$
    with $\Sigma^* \in \DLang{\V}(\Sigma)$. Thus, the singleton
    language $\set{w}$ belongs to $\LVariety{C}(\Sigma)$ and since this is true
    for any $w \in F$ and $F$ is finite, we can deduce from this that $F$ is in
    $\LVariety{C}(\Sigma)$, as the latter is trivially closed under Boolean
    operations.

    Now, for all $m \in N$, the language $\mu^{-1}(m)$ belongs to
    $\DLang{\V}(\Sigma)$, so we finally have $L \in \LVariety{C}(\Sigma)$.
    This is true for any $L \in \DLang{\StVEsntl\V}(\Sigma)$, so in conclusion,
    $\DLang{\StVEsntl\V}(\Sigma) \subseteq \LVariety{C}(\Sigma)$.
\end{proof}

Proposition~\ref{ptn:Essentially-V-V_join_LI} then follows from the two next
lemmata, that are both easy consequences of
Proposition~\ref{ptn:Essentially-V_languages}. For completeness, we give the
proofs in the appendix.

\begin{lemma}
\label{lem:Join_is_essentially-V}
    Let $\V$ be a variety of monoids.
    Then $\StVGen{\FSVVjoinLI}{\nem} \subseteq \StVEsntl\V$.
\end{lemma}

\begin{lemma}
\label{lem:Conditional_essentially-V_is_join}
    Let $\V$ be a variety of monoids.
    Then $\StVEsntl\V \subseteq \StVGen{\FSVVjoinLI}{\nem}$ if and only if $\V$
    verifies criterion~\ref{ctn:variety_criterion}.
\end{lemma}

\section{Applications}
\label{sec:Applications}

In this last section, we use the link between essentially-$\V$ stamps and
$\FSVVjoinLI$ to reprove some characterisations of joins between $\FSVLI$ and
some well-known varieties of monoids in terms of identities.

One thing seems at first glance a bit problematic about proving that a variety
of monoids $\V$ satisfies criterion~\ref{ctn:variety_criterion}. Indeed, to this
end, one needs to prove that certain languages belong to $\DLang{\FSVVjoinLI}$;
however, this poses a problem when one's goal is precisely to characterise
$\FSVVjoinLI$, because one shall a priori not know more about
$\DLang{\FSVVjoinLI}$ than what is given by
Proposition~\ref{ptn:V_join_LI_ne-variety}.
Nevertheless, there is a natural sufficient condition for
criterion~\ref{ctn:variety_criterion} to hold that depends only on $\DLang{\V}$:
if given any language $L \in \DLang{\V}(\Sigma)$ and any $x, y \in \Sigma^*$
with $\Sigma$ an alphabet, there exists a language $K \in \DLang{\V}(\Sigma)$
such that $L$ is equal to the quotient $x^{-1} K y^{-1}$, then $\V$ verifies
criterion~\ref{ctn:variety_criterion}. We don't know whether this
quotient-expressibility condition that solely depends on the variety $\V$
(without explicit reference to $\FSVLI$) is actually equivalent to it satisfying
criterion~\ref{ctn:variety_criterion}, but we can prove such an equivalence for
a weaker quotient-expressibility condition for $\V$. The proof is to be found in
the appendix.

\begin{proposition}
\label{ptn:Variety_criterion}
    Let $\V$ be a variety of monoids.
    Then $\V$ satisfies criterion~\ref{ctn:variety_criterion} if and only if for
    any $L \in \DLang{\V}(\Sigma)$ and any $x, y \in \Sigma^*$ with $\Sigma$ an
    alphabet, there exist $k, l \in \N$ such that for all
    $u \in \Sigma^k, v \in \Sigma^l$, there exists a language
    $K \in \DLang{\V}(\Sigma)$ verifying
    $u^{-1} L v^{-1} = (x u)^{-1} K (v y)^{-1}$.
\end{proposition}

This quotient-expressibility condition appears to be particularly useful to
prove that a variety of monoids $\V$ does not satisfy
criterion~\ref{ctn:variety_criterion} without needing to understand what
$\DLang{\FSVVjoinLI}$ is. We demonstrate this for the variety of finite
commutative and idempotent monoids $\FMVJ[1]$.

\begin{proposition}
    $\FMVJ[1]$ does not satisfy criterion~\ref{ctn:variety_criterion}.
\end{proposition}

\begin{proof}
    Given an alphabet $\Sigma$, the set $\DLang{\FMVJ[1]}(\Sigma)$ consists of
    all Boolean combinations of languages of the form $\Sigma^* a \Sigma^*$ for
    $a \in \Sigma$ (see~\cite[Chapter~2, Proposition~3.10]{Books/Pin-1986}).

    Let $L = \set{a, b}^* b \set{a, b}^* \in \DLang{\FMVJ[1]}(\set{a, b})$ and
    $x = b, y = \emptyword$.
    Take any $k, l \in \N$ and set $u = a^k$ and $v = a^l$.
    Consider a $K \in \DLang{\FMVJ[1]}(\set{a, b})$. We have that
    $x u a v y \in K \Leftrightarrow x u a b v y \in K$ so that
    $a \in (x u)^{-1} K (v y)^{-1} \Leftrightarrow
     a b \in (x u)^{-1} K (v y)^{-1}$.
    But $a \notin u^{-1} L v^{-1}$ and $a b \in u^{-1} L v^{-1}$, hence
    $u^{-1} L v^{-1} \neq (x u)^{-1} K (v y)^{-1}$ and this holds for any choice
    of $K$.
    So for any $k, l \in \N$, there exists $u \in \Sigma^k, v \in \Sigma^l$ such
    that no $K \in \DLang{\FMVJ[1]}(\set{a, b})$ verifies
    $u^{-1} L v^{-1} = (x u)^{-1} K (v y)^{-1}$.

    In conclusion, by Proposition~\ref{ptn:Variety_criterion}, $\FMVJ[1]$ does
    not satisfy criterion~\ref{ctn:variety_criterion}.
\end{proof}

We now prove the announced characterisations of joins between $\FSVLI$ and some
well-known varieties of monoids in terms of identities.

\begin{theorem}
    We have the following.
    \begin{enumerate}
	\item
	\label{thm:R_join_LI}
	    $\StVGen{\FMVR \FSVjoin \FSVLI}{\nem} = \StVEsntl\FMVR =
	     \StVidentities{x^\omega y (a b)^\omega a z t^\omega =
			    x^\omega y (a b)^\omega z t^\omega}_\nem$.
	\item
	\label{thm:L_join_LI}
	    $\StVGen{\FMVL \FSVjoin \FSVLI}{\nem} = \StVEsntl\FMVL =
	     \StVidentities{x^\omega y b (a b)^\omega z t^\omega =
			    x^\omega y (a b)^\omega z t^\omega}_\nem$.
	\item
	\label{thm:J_join_LI}
	    $\StVGen{\FMVJ \FSVjoin \FSVLI}{\nem} = \StVEsntl\FMVJ =
	     \StVidentities{x^\omega y (a b)^\omega a z t^\omega =
			    x^\omega y (a b)^\omega z t^\omega,
			    x^\omega y b (ab)^\omega z t^\omega =
			    x^\omega y (a b)^\omega z t^\omega}_\nem$.
	\item
	\label{thm:H_join_LI}
	    $\StVGen{\FMVariety{H} \FSVjoin \FSVLI}{\nem} =
	     \StVEsntl\FMVariety{H}$
	    for any variety of groups $\FMVariety{H}$.
    \end{enumerate}
\end{theorem}

\begin{proof}
    In each case, we prove that the variety of monoids under consideration
    satisfies criterion~\ref{ctn:variety_criterion} using
    Proposition~\ref{ptn:Variety_criterion}. We then use
    Propositions~\ref{ptn:Essentially-V-V_join_LI}~and~\ref{ptn:Essentially-V_identities}.

    \proofsubparagraph{Proof of~\ref{thm:R_join_LI}.}
    It is well-known that given an alphabet $\Sigma$, the set
    $\DLang{\FMVR}(\Sigma)$ consists of all languages that are disjoint unions
    of languages that are of the form $A_0^* a_1 A_1^* \cdots a_k A_k^*$ where
    $k \in \N$, $a_1, \ldots, a_k \in \Sigma$,
    $A_0, A_1, \ldots, A_k \subseteq \Sigma$ and $a_i \notin A_{i - 1}$ for
    all $i \in [k]$ (see~\cite[Chapter~4, Theorem~3.3]{Books/Pin-1986}).
    
    Let $\Sigma$ be an alphabet and take a language
    $A_0^* a_1 A_1^* \cdots a_k A_k^*$ where $k \in \N$,
    $a_1, \ldots, a_k \in \Sigma$, $A_0, A_1, \ldots, A_k \subseteq \Sigma$ and
    $a_i \notin A_{i - 1}$ for all $i \in [k]$.
    Take $x, y \in \Sigma^*$. Observe that $y$ can be uniquely written as
    $y = z t$ where $z \in A_k^*$ and
    $t \in \set{\emptyword} \cup (\Sigma \setminus A_k) \Sigma^*$.
    We have
    \[
	A_0^* a_1 A_1^* \cdots a_k A_k^*
	= x^{-1}
	  \Bigl(x A_0^* a_1 A_1^* \cdots a_k A_k^* t \cap
	        \bigcap_{v \in A_k^{< \length{z}}}
	        (\Sigma^* \setminus x A_0^* a_1 A_1^* \cdots a_k v t)\Bigr)
	  y^{-1}
    \]
    using the convention that $x A_0^* a_1 A_1^* \cdots a_k v t = x v t$ for all
    $v \in A_k^{< \length{z}}$ when $k = 0$.
    The language
    $x A_0^* a_1 A_1^* \cdots a_k A_k^* t \cap
     \bigcap_{v \in A_k^{< \length{z}}}
     (\Sigma^* \setminus x A_0^* a_1 A_1^* \cdots a_k v t)$
    does belong to the set $\DLang{\FMVR}(\Sigma)$ because the latter is closed
    under Boolean operations and by definition of $z$ and $t$.
    Thus, we can conclude that for each $L \in \DLang{\FMVR}(\Sigma)$ and
    $x, y \in \Sigma^*$, there exists $K \in \DLang{\FMVR}(\Sigma)$ such that
    $L = x^{-1} K y^{-1}$ by using the characterisation of
    $\DLang{\FMVR}(\Sigma)$, the fact that quotients commute with
    unions~\cite[p.~20]{Books/Pin-1986} and closure of $\DLang{\FMVR}(\Sigma)$
    under unions.

    \proofsubparagraph{Proof of~\ref{thm:L_join_LI}.}
    It is also well-known that given an alphabet $\Sigma$, the set
    $\DLang{\FMVL}(\Sigma)$ consists of all languages that are disjoint unions
    of languages that are of the form $A_0^* a_1 A_1^* \cdots a_k A_k^*$ where
    $k \in \N$, $a_1, \ldots, a_k \in \Sigma$,
    $A_0, A_1, \ldots, A_k \subseteq \Sigma$ and $a_i \notin A_i$ for all
    $i \in [k]$ (see~\cite[Chapter~4, Theorem~3.4]{Books/Pin-1986}).
    The proof is then dual to the previous case.

    \proofsubparagraph{Proof of~\ref{thm:J_join_LI}.}
    Given an alphabet $\Sigma$, for each $k \in \N$, we define the equivalence
    relation $\sim_k$ on $\Sigma^*$ by $u \sim_k v$ for $u, v \in \Sigma^*$
    whenever $u$ and $v$ have the same set of subwords of length at most $k$.
    This relation is a congruence of finite index on $\Sigma^*$.
    Simon proved~\cite{Simon-1975} that a language belongs to
    $\DLang{\FMVJ}(\Sigma)$ if and only it is equal to a union of
    $\sim_k$-classes for a $k \in \N$.

    Let $\Sigma$ be an alphabet and take $L \in \DLang{\FMVJ}(\Sigma)$ as well
    as $x, y \in \Sigma^*$. Thus, there exists $k \in \N$ such that $L$ is a
    union of $\sim_k$-classes. Define the language
    $K = \bigcup_{w \in L} [x w y]_{\sim_{\length{x y} + k}}$: it belongs to
    $\DLang{\FMVJ}(\Sigma)$ by construction.
    We now show that $L = x^{-1} K y^{-1}$, which concludes the proof.
    Let $w \in L$: we have that
    $x w y \in [x w y]_{\sim_{\length{x y} + k}} \subseteq K$, so that
    $w \in x^{-1} K y^{-1}$.
    Let conversely $w \in x^{-1} K y^{-1}$. This means that $x w y \in K$, which
    implies that there exists $w' \in L$ such that
    $x w y \sim_{\length{x y} + k} x w' y$. Actually, it holds that any
    $u \in \Sigma^*$ of length at most $k$ is a subword of $w$ if and only if it
    is a subword of $w'$, because $x u y$ is a subword of $x w y$ if and only if
    it is a subword of $x w' y$. Hence, $w \sim_k w'$, which implies that
    $w \in L$.

    \proofsubparagraph{Proof of~\ref{thm:H_join_LI}.}
    Consider any variety of groups $\FMVariety{H}$.
    Take a language $L \in \DLang{\FMVariety{H}}(\Sigma)$ for an alphabet
    $\Sigma$ and let $x, y \in \Sigma^*$.
    Consider the syntactic morphism $\eta\colon \Sigma^* \to M$ of $L$: we have
    that $M$ is a group in $\FMVariety{H}$. Define the language
    $K = \eta^{-1}\bigl(\eta(x) \eta(L) \eta(y)\bigr)$: it belongs to
    $\DLang{\FMVariety{H}}(\Sigma)$.
    We now show that $L = x^{-1} K y^{-1}$, which concludes the proof.
    Let $w \in L$: we have that $\eta(x w y) \in \eta(x) \eta(L) \eta(y)$, so
    that $w \in x^{-1} K y^{-1}$.
    Conversely, let $w \in x^{-1} K y^{-1}$. We have that $x w y \in K$, which
    means that $\eta(x w y) = \eta(x) \eta(w') \eta(y)$ for a $w' \in L$, so
    that $\eta(w) = \eta(w') \in \eta(L)$, as any element in $M$ is invertible.
    Thus, $w \in L$.
\end{proof}

\section{Conclusion}

The general method presented in this paper actually allows to reprove in a
straightforward language-theoretic way even more characterisations of the join
of $\FSVLI$ with some variety of finite monoids. This can for instance be done
for the variety of finite commutative monoids $\FMVCom$ or the variety of finite
commutative aperiodic monoids $\FMVACom$.

In fact, as already observed in some sense by Costa~\cite{Costa-2001}, many
varieties of finite monoids seem to verify
criterion~\ref{ctn:variety_criterion}. The main question left open by this
present work is to understand better what exactly those varieties are. Another
question left open is whether Proposition~\ref{ptn:Variety_criterion} can be
refined by using the stronger quotient-expressibility condition alluded to
before the statement of the proposition. The answers to both questions are
unclear to the author, but making progress on them may also lead to a better
understanding of joins of varieties of finite monoids with $\FSVLI$.

\bibliographystyle{plainurl}
\bibliography{Bibliography}

\appendix
\section{Missing proofs}

\begin{proof}[Proof of Proposition~\ref{ptn:V_join_LI_ne-variety}]
    Let $\W$ be an \nem-variety of stamps such that 
    $\StVGen{\V}{\allm} \cup \StVGen{\FSVLI}{\nem} \subseteq \W$.
    There exists a variety of semigroups $\FSVariety{W'}$ such that
    $\StVGen{\FSVariety{W'}}{\nem} = \W$.

    Let $S \in \V \cup \FSVLI$. We denote by $S^1$ the monoid $S$ if $S$ is
    already a monoid and the monoid $S \cup \set{1}$ otherwise. Then the
    evaluation morphism $\eta_S\colon S^* \to S^1$ such that $\eta_S(s) = s$ for
    all $s \in S$ verifies $\eta_S(S^+) = S$ and additionally $S^1 = S$ when
    $S \in \V$. This implies that
    $\eta_S \in \StVGen{\V}{\allm} \cup \StVGen{\FSVLI}{\nem} \subseteq \W$. But
    by definition of $\FSVariety{W'}$, it must be that
    $S = \eta_S(S^+) \in \FSVariety{W'}$.

    Therefore, $\FSVariety{W'}$ contains both $\V$ and $\FSVLI$, which implies
    that $\FSVVjoinLI \subseteq \FSVariety{W'}$ by inclusion-wise minimality of
    $\FSVVjoinLI$. By definition, we can then conclude that
    $\StVGen{\FSVVjoinLI}{\nem} \subseteq \StVGen{\FSVariety{W'}}{\nem} = \W$.
    So $\StVGen{\FSVVjoinLI}{\nem}$ is the inclusion-wise least \nem-variety of
    stamps containing both $\StVGen{\V}{\allm}$ and $\StVGen{\FSVLI}{\nem}$.

    \medskip

    Let now $\LVariety{W}$ be an \nem-variety of languages such that
    $\DLang{\V} \cup \DLang{\FSVLI} \subseteq \LVariety{W}$. It holds that
    $\LVariety{W} = \DLang{\W}$ for an \nem-variety of stamps $\W$.
    We have that $\StVGen{\V}{\allm}$, which is in particular an \nem-variety of
    stamps, is included in $\W$ because
    $\DLang{\StVGen{\V}{\allm}} = \DLang{\V} \subseteq \LVariety{W} =
     \DLang{\W}$,
    but also that $\StVGen{\FSVLI}{\nem}$ is included in $\W$ because
    $\DLang{\StVGen{\FSVLI}{\nem}} = \DLang{\FSVLI} \subseteq
     \LVariety{W} = \DLang{\W}$.
    By inclusion-wise minimality of $\StVGen{\FSVVjoinLI}{\nem}$, it follows
    that $\StVGen{\FSVVjoinLI}{\nem} \subseteq \W$. Hence, using again the above
    fact on the Eilenberg correspondence, we can conclude that
    $\DLang{\FSVVjoinLI} =
     \DLang{\StVGen{\FSVVjoinLI}{\nem}} \subseteq \DLang{\W} = \LVariety{W}$.
    So $\DLang{\FSVVjoinLI}$ is the inclusion-wise least \nem-variety of
    languages containing both $\DLang{\V}$ and $\DLang{\FSVLI}$.

    Consider now the class of languages $\LVariety{C}$ such that
    $\LVariety{C}(\Sigma)$ is the Boolean closure of
    $\DLang{\V}(\Sigma) \cup \DLang{\FSVLI}(\Sigma)$ for each alphabet $\Sigma$.
    By closure under Boolean operations of $\DLang{\FSVVjoinLI}$, we have that
    $\LVariety{C} \subseteq \DLang{\FSVVjoinLI}$.
    Now, as Boolean operations commute with both
    quotients~\cite[p.~20]{Books/Pin-1986} and inverses of
    \nem-morphisms~\cite[Proposition~0.4]{Books/Pin-1986}, by closure of
    $\DLang{\V}$ and $\DLang{\FSVLI}$ under quotients and inverses of
    \nem-morphisms, we actually have that $\LVariety{C}$ is an \nem-variety of
    languages. Therefore, by inclusion-wise minimality of $\DLang{\FSVVjoinLI}$,
    we can conclude that $\DLang{\FSVVjoinLI} = \LVariety{C}$.
\end{proof}

\begin{proof}[Proof of Lemma~\ref{lem:Join_is_essentially-V}]
    We actually have that
    $\DLang{\V} \cup \DLang{\FSVLI} \subseteq \DLang{\StVEsntl\V}$, which
    allows us to conclude by inclusion-wise minimality of
    $\DLang{\FSVVjoinLI}$ (Proposition~\ref{ptn:V_join_LI_ne-variety})
    and by the fact that $\DLang{\StVEsntl\V}$ is an \nem-variety of languages
    (Proposition~\ref{ptn:Essentially-V_identities}).

    Let $\Sigma$ be an alphabet.
    The fact that $\DLang{\V}(\Sigma) \subseteq \DLang{\StVEsntl\V}(\Sigma)$
    follows trivially from Proposition~\ref{ptn:Essentially-V_languages}.
    Moreover, for all $u \in \Sigma^*$, since necessarily
    $\Sigma^* \in \DLang{\V}(\Sigma)$, we have that both $u \Sigma^*$ and
    $\Sigma^* u$ belong to $\DLang{\FSVLI}(\Sigma)$. Thus, as
    $\DLang{\StVEsntl\V}(\Sigma)$ is closed under Boolean operations, it follows
    that $\DLang{\FSVLI}(\Sigma) \subseteq \DLang{\StVEsntl\V}(\Sigma)$.

    This concludes the proof, since it holds for any alphabet $\Sigma$.
\end{proof}

\begin{proof}[Proof of Lemma~\ref{lem:Conditional_essentially-V_is_join}]
    Assume that $\StVEsntl\V \subseteq \StVGen{\FSVVjoinLI}{\nem}$.
    For any $L \in \DLang{\V}(\Sigma)$ and any $x, y \in \Sigma^*$ with $\Sigma$
    an alphabet, by Proposition~\ref{ptn:Essentially-V_languages}, we have that
    $x L y \in \DLang{\StVEsntl\V}(\Sigma) \subseteq
     \DLang{\FSVVjoinLI}(\Sigma)$.
    Hence, $\V$ verifies criterion~\ref{ctn:variety_criterion}.

    Conversely, assume that $\V$ verifies criterion~\ref{ctn:variety_criterion}.
    For any alphabet $\Sigma$, the set $\DLang{\FSVVjoinLI}(\Sigma)$ contains
    all languages of the form $x L y$ for $L \in \DLang{\V}(\Sigma)$ and
    $x, y \in \Sigma^*$, so it contains all Boolean combinations of languages of
    that form, since it is closed under Boolean operations.
    Therefore, by Proposition~\ref{ptn:Essentially-V_languages}, we have
    $\DLang{\StVEsntl\V} \subseteq \DLang{\FSVVjoinLI}$, so that
    $\StVEsntl\V \subseteq \StVGen{\FSVVjoinLI}{\nem}$.
\end{proof}

\begin{proof}[Proof of Proposition~\ref{ptn:Variety_criterion}]
    Let us first observe that given any alphabet $\Sigma$, given any language
    $K$ on that alphabet and given any two words $x, y \in \Sigma^*$, we have
    that $x (x^{-1} K y^{-1}) y = x \Sigma^* y \cap K$ and
    $x^{-1} (x K y) y^{-1} = K$.

    \proofsubparagraph{Implication from right to left.}
    Assume that for any $L \in \DLang{\V}(\Sigma)$ and any $x, y \in \Sigma^*$
    with $\Sigma$ an alphabet, there exist $k, l \in \N$ such that for all
    $u \in \Sigma^k, v \in \Sigma^l$, there exists a language
    $K \in \DLang{\V}(\Sigma)$ verifying
    $u^{-1} L v^{-1} = (x u)^{-1} K (v y)^{-1}$.
    Take $L \in \DLang{\V}(\Sigma)$ for an alphabet $\Sigma$ and take
    $x, y \in \Sigma^*$. Consider also $k, l \in \N$ that are guaranteed to
    exist by the assumption we just made.

    For all $u \in \Sigma^k, v \in \Sigma^l$, there exists a language
    $K \in \DLang{\V}(\Sigma)$ verifying
    $u^{-1} L v^{-1} = (x u)^{-1} K (v y)^{-1}$, so that by our observation at
    the beginning of the proof, we have
    \[
	x (u \Sigma^* v \cap L) y = x u (u^{-1} L v^{-1}) v y =
	x u \bigl((x u)^{-1} K (v y)^{-1}\bigr) v y =
	x u \Sigma^* v y \cap K
	\displaypunct{.}
    \]
    Using Proposition~\ref{ptn:V_join_LI_ne-variety}, we thus have that
    $x (u \Sigma^* v \cap L) y \in \DLang{\V \FSVjoin \FSVLI}(\Sigma)$ for all
    $u \in \Sigma^k, v \in \Sigma^l$.
    Moreover, since we have that the set of words of $L$ of length at least
    $k + l$ is
    \[
	\Sigma^{\geq k + l} \cap L =
	\bigcup_{u \in \Sigma^k, v \in \Sigma^l} (u \Sigma^* v \cap L)
    \]
    and since
    \[
	L = (\Sigma^{\geq k + l} \cap L) \cup F
    \]
    where $F$ is a finite set of words on $\Sigma$ of length less than $k + l$,
    we have that
    \[
	x L y = x \bigl((\Sigma^{\geq k + l} \cap L) \cup F\bigr) y =
	\bigcup_{u \in \Sigma^k, v \in \Sigma^l} x (u \Sigma^* v \cap L) y \cup
	x F y
	\displaypunct{.}
    \]
    We can thus conclude that $x L y \in \DLang{\FSVVjoinLI}(\Sigma)$ since
    $x F y \in \DLang{\FSVLI}(\Sigma)$ and because $\DLang{\FSVVjoinLI}(\Sigma)$
    is closed under unions.

    \proofsubparagraph{Implication from left to right.}
    Assume that $\V$ satisfies criterion~\ref{ctn:variety_criterion}.
    Take $L \in \DLang{\V}(\Sigma)$ for an alphabet $\Sigma$ and take
    $x, y \in \Sigma^*$. By hypothesis, we know that
    $x L y \in \DLang{\FSVVjoinLI}(\Sigma)$.

    By Proposition~\ref{ptn:V_join_LI_ne-variety}, this means that $x L y$ is a
    Boolean combination of languages in
    $\DLang{\V}(\Sigma) \cup \DLang{\FSVLI}(\Sigma)$. Further, this implies that
    $x L y$ can be written as the union of intersections of languages of
    $\DLang{\V}(\Sigma)$ and $\DLang{\FSVLI}(\Sigma)$ or their complements,
    which in turn implies, by closure of $\DLang{\V}(\Sigma)$ and
    $\DLang{\FSVLI}(\Sigma)$ under Boolean operations, that $x L y$ can be
    written as a finite union of languages of the form
    $K \cap (U \Sigma^* V \cup W)$ with $K \in \DLang{\V}(\Sigma)$ and
    $U, V, W \subseteq \Sigma^*$ finite. Since any word in $x L y$ must be of
    length at least $\length{x y}$ and have $x$ as a prefix and $y$ as a suffix,
    we can assume that any language $K \cap (U \Sigma^* V \cup W)$ appearing in
    a finite union as described above verifies that $U \subseteq x \Sigma^*$,
    that $V \subseteq \Sigma^* y$ and that $W \subseteq x \Sigma^* y$.
    Now, if we take $k, l \in \N$ big enough, we thus have that
    \[
	x L y =
	\bigcup_{u \in \Sigma^k, v \in \Sigma^l}
	(K_{u, v} \cap x u \Sigma^* v y) \cup F
    \]
    where $K_{u, v} \in \DLang{\V}(\Sigma)$ for all
    $u \in \Sigma^k, v \in \Sigma^l$ and
    $F \subseteq \Sigma^{< \length{x y} + k + l}$.
    Hence, for all $u \in \Sigma^k, v \in \Sigma^l$, we have
    \begin{align*}
	u^{-1} L v^{-1}
	& = u^{-1} \bigl(x^{-1} (x L y) y^{-1}\bigr) v^{-1}\\
	& = (x u)^{-1}
	    \Bigl(\bigcup_{u' \in \Sigma^k, v' \in \Sigma^l}
		  (K_{u', v'} \cap x u' \Sigma^* v' y) \cup F\Bigr)
	    (v y)^{-1}\\
       	& = \begin{aligned}[t]
		& \bigcup_{u' \in \Sigma^k, v' \in \Sigma^l}
		  (x u)^{-1}
		  \Bigl(x u'
			\bigl((x u')^{-1} K_{u', v'} (v' y)^{-1}\bigr)
			v' y\Bigr)
		  (v y)^{-1} \cup\\
		& (x u)^{-1} F (v y)^{-1}
	    \end{aligned}\\
	& = (x u)^{-1} K_{u, v} (v y)^{-1}
	\displaypunct{,}
    \end{align*}
    using classical formulae for quotients~\cite[p.~20]{Books/Pin-1986} and
    observing that $(x u)^{-1} K (v y)^{-1} = \emptyset$ for any
    $K \subseteq \Sigma^*$ such that $K \cap x u \Sigma^* v y = \emptyset$.
\end{proof}
 
\end{document}